\documentclass[12pt,a4paper]{article}
\usepackage[T1]{fontenc}
\usepackage[utf8]{inputenc}
\usepackage[dvipsnames]{xcolor}
\usepackage[english]{babel}
\usepackage{manfnt}
\usepackage[numbers]{natbib}
\usepackage[T1,hyphens]{url}
\usepackage{hyperref}
\hypersetup{colorlinks=true, linkcolor=blue, urlcolor=blue, citecolor=blue,  filecolor=magenta,   pdfpagemode=FullScreen,}
\usepackage{amsfonts, amsmath, amssymb}  

\newtheorem{definition}{Definition}
\newtheorem{theorem}{Theorem}

\newtheorem{lemma}{Lemma}

\newtheorem{corollary}{Corollary}

\newenvironment{proof}{\textsf{\bfseries{}Proof. }}{\qed\par}

 \newcommand{\bbbn}{\mathbb{N}}

 \newcommand{\dom}{\mathrm{dom}}
 
 \newcommand{\x}{\mathbf{x}}

 \newcommand{\e}{\in}
 \newcommand{\s}{\rightarrow}
 \newcommand{\y}{\mathbf{y}}

\newcommand{\rb}{\raisebox{0.5ex}}
\usepackage{verbatim}
\usepackage{enumerate}

\usepackage{physics}

 \newcommand{\iprod}[2]{\langle #1 | #2 \rangle}
 
 \newcommand{\oprod}[2]{| #1 \rangle\langle #2 |}
 \newcommand{\qed}{\mbox{ }\hfill$\square$\par}

\usepackage{times}
\usepackage{sectsty}
\usepackage{xcolor,colortbl}
\usepackage[framemethod=tikz]{mdframed}
\usepackage{caption}
\usepackage{adjustbox}
\usepackage{systeme}
\usepackage{multicol}

\begin{document}

\title{A New Quantum Random Number Generator Certified by Value Indefiniteness}

 \author{Jos\'{e} Manuel Ag\"{u}ero
	 Trejo and Cristian S. Calude\\
	 School of Computer Science\\ University of Auckland, New Zealand
	 }
	 \maketitle
\begin{abstract}
In this paper we propose a new ternary QRNG based on measuring located value indefinite observables with probabilities $1/4, 1/2, 1/4$ and prove that every sequence generated  is maximally unpredictable,  3-bi-immune (a stronger form of bi-immunity),  and its prefixes are Borel normal.  The  ternary quantum random digits produced by the   QRNG are algorithmically transformed  into quantum random bits using an alphabetic morphism which preserves all the above properties. 
\end{abstract}

\section{Introduction}

\label{intro}
Randomness is an important resource in science, statistics, cryptography, gambling, medicine, art and politics. Pseudo-random number generators (PRNGs) -- computer algorithms designed to simulate randomness -- have been  the main, if not  the only,  sources of randomness for a long time, but their quality is weak.
As early as 1951 von Neumann realised the danger of 
mistakenly believing that PRNGs produce  ``true`` 
randomness~\cite{von-neumann1}: ``Anyone who attempts to generate random numbers by deterministic means is, of course, living in a state of sin.'' Problems with the poor quality  PRNGs are well known: a classical example is  the discovery in 2012 of a weakness in a worldwide-used  encryption system which  was traced to  a PRNG~\cite{factor_wrong2012}.

With the development of algorithmic information theory~\cite{ch6,li-vitanyi-2019,DH}  various classes  of (algorithmic) random strings/sequences have been studied and von  Neumann intuition was rigorously proved in a more general form: {\em mathematically there is no 
`true`` random string/sequence}~\cite{calude:02}.

The importance of high quality randomness -- which is obvious in cryptography, where good randomness is vital to the security of data and communication, but is equally true in other areas ranging from statistics and information science to   medicine, physics, politics and religion -- 
  has driven a recent surge of interest 
 in developing  ``better than PRNG" random number generators, in particular, quantum random number generators (QRNGs)~\cite{Naco_rand,RevModPhys.89.015004}.
QRNGs are generally considered to be, by their very nature, ``better than PRNGs" and are expected to
``excel" precisely on properties of randomness where algorithmic PRNGs obviously fail: incomputability and inherent unpredictability.  The formulation {\it ``better than PRNGs"} can be read into two radically different ways: a) ``better'' than {\em some} PRNGs, b) ``better'' than {\em any} PRNGs. Of course, b) is the required property.

To date only one class of QRNGs has been proved to satisfy b)~\cite{Abbott:2010uq,acs-2015-info6040773,PhysRevLett.119.240501}. This type of QRNG is based on a located form~\cite{abbott2012strongrandomness,Abbott:2013ly, DBLP:conf/birthday/AbbottCS15,2015-AnalyticKS} of  the Kochen-Specker Theorem~\cite{Kochen:2017aa}, a result true only in Hilbert spaces of dimension {\em at least three}. These QRNGs -- which locate and repeatedly measure a value-indefinite quantum observable -- produce more than incomputable sequences (over alphabets with at least three letters);  more precisely, they generate sequences having a form of algorithmic randomness called {\it  bi-immunity}~\cite{DH},   that is, sequences for which no algorithm can compute more than finitely many exact values.
The experimental analysis of 10 samples of $2^{30}$ binary strings generated 
with the implementation~\cite{PhysRevLett.119.240501} of the QRNG proposed in~\cite{Abbott:2010uq,acs-2015-info6040773} showed incomputability in a weak and not decisive manner.  Some possible reasons include a problematic branch with probability zero used in the generalised beam splitter  -- recall,  the Kochen-Specker Theorem is false in dimension 2 --, the  not long enough length of samples, and, of course, imperfections in the implementation of the measuring protocol~\cite{AbbottCalude10}.

In this paper we improve the QRNG~\cite{Abbott:2010uq,acs-2015-info6040773,PhysRevLett.119.240501} and propose a new ternary QRNG based on measuring located value indefinite observables with probabilities $1/4, 1/2,1/4$. {\em We prove that every sequence generated  is maximally unpredictable,  3-bi-immune,  and its prefixes are Borel normal.}  The  ternary quantum random digits produced by the   QRNG are algorithmically transformed  into quantum random bits using an alphabetic morphism which preserves all the above properties.

The paper is organised as follows. Section~\ref{notat} includes the notation and main definitions.
 In Section~\ref{val_in} we  present the main theoretical basis of the QRNG: localising value indefinite observables, 
and their unpredictability.  Section~\ref{qrng1} is devoted to the blueprint of the original QRNG based on Spin-1; in 
Section~\ref{qrng2} we present the new QRNG. In Section~\ref{3} we prove the main  properties of ternary sequences
produced by the QRNG and in  Section~\ref{2} we introduce the transformation from ternary to binary and prove that
it preserves all properties proved in the previous section. The last section includes a summary and further questions.

\section{Notation and definitions}
\label{notat}
The set of positive integers will be denoted by $\bbbn$. Consider the alphabet
$A_{b}=\{0,1,\dots,b-1\}$, where $b\ge 2$ is an integer; the elements of
$A_b$ are to be considered the digits used in natural positional
representations of numbers in the interval $[0,1)$ at base $b$. By  $A_{b}^{*}$
and  $A_{b}^{\omega}$ we denote the sets of (finite)  strings and (infinite)
sequences over the alphabet $A_{b}$. Strings will be denoted by $x,y,u,w$; the
length of the string $x= x_1x_2\dots x_m$, $x_i\in A_{b}$, is denoted by
$|x|_{b}=m$ (the subscript $b$ will be omitted if it is clear from the
context); $A_{b}^{m}$ is the set of all strings of length $m$.  Sequences will
be denoted by $\mathbf{x}= x_1x_2\dots$; the prefix of length $m$ of
$\mathbf{x}$ is the string $ \mathbf{x}(m)= x_1x_2\dots x_m$. 
Strings
will be ordered quasi-lexicographically according to the natural order $0<1<2  < \dots <b-1$ on the
alphabet $A_{b}$. For example, for $b=2$, we have $0<1<00<01<10<11<000 \dots$. 
We
assume knowledge of elementary computability theory over different size
alphabets~\cite{calude:02}. 
Sequences can be also viewed as $A_{b}$-valued
functions defined on $\bbbn$.

%
%
%

%

Let $\mathcal{B}(\mathbb{R})$ be the class of Borel sets in $\mathbb{R}$, that is, the smallest $\sigma$-algebra containing all opens sets.
 Let $(\Omega,\mathcal{F},\mathbb{P})$ be a probability space. A random variable  $X:\Omega\to\mathbb{R}$ is a function such that for every $B\in\mathcal{B}(\mathbb{R})$ we have $\{w\in\Omega:X(w)\in B\}\in\mathcal{F}$. Furthermore, if for all $x,y\in\mathbb{R}$,  $\mathbb{P}(X\leq x,Y\leq y)=\mathbb{P}(X\leq x)\mathbb{P}(Y\leq y)$, we say that $X,Y$ are independent random variables.  By $\mathbb{E}(X)= \sum_{x} x\mathbb{P}(X=x)$
 we denote the expectation of the random variable 
 $X$~\cite{Billingsley:1979aa}. Let $u\in A_b^{*}$, $uA_b^{\omega} = \{\x \in A_b^{\omega} : \x(|u|)= u\}$ and consider the smallest $\sigma$-algebra $\mathcal{B}(A_b^{\omega})$ generated by the family
($uA_b^{\omega} : u\in A_b^{*}$). The {\it Lebesgue space (probability) }is the probability space $(A_b^{\omega},\mathcal{B}(A_2^{\omega}),\mathbb{P})$ where $\mathbb{P}(uA_b^{\omega}) = b^{-|u|}$~\cite{calude:02}.

In contrast to the bounds on probability distributions given by Bell Theorem~\cite{bell-66,bell-87} under the  premise of locality, Kochen-Specker Theorem shows that, assuming  non-contextuality\footnote{Informally, the property that the outcome of the measurement of a quantum observable is independent  of how that value is eventually measured.}, the Hilbert-space structure of quantum mechanics makes it  impossible to assign ``classical'' definite values to all possible quantum observables in a consistent manner. 
Since such a definite value is precisely a (deterministic) hidden variable specifying, in advance, the result of a measurement of an observable, the theorem shows that the outcomes of all quantum measurements on a system cannot be simultaneously pre-determined.

As is common in modern treatments of the Kochen-Specker Theorem~\cite{Cabello:1994ly,Cabello:1996zr,peres-91}  we focus on one-dimensional (rank-1) projection observables, and we denote the observable projecting onto the linear subspace spanned by a vector $\ket{\psi}$ as $P_\psi=\frac{\oprod{\psi}{\psi}}{|\iprod{\psi}{\psi}|}$\rb.
We then fix a positive integer $n\geq 2$ and let $O \subseteq \{  P_{\psi}:  \ket{\psi} \in \mathbb{C}^{n} \}$ be a non-empty set of one-dimensional projection observables on the Hilbert space $\mathbb{C}^{n}$.

\begin{definition}
A set $C \subset O$ is a \emph{context} of $O$ if $C$ has $n$ elements (i.e $|C| = n$) and for all $P_{\psi}, P_{\phi} \in C$ with $P_{\psi} \neq P_{\phi}, \braket{\psi}{\phi} = O$. 
\end{definition}

\begin{definition}
A  \emph{value assignment function} (on $O$) is a partial function $v:O\to \{0,1\}$ assigning values to some (possibly all) observables in $O$.  The partiality of the  function $v$ means that $v(P)$ can be $0,1$ or indefinite.
\end{definition}

\begin{definition}
An observable $P \in O$ is  \emph{value definite} (under the assignment function $v$)  if $v(P)$ is defined, i.e.~it is 0 or 1; otherwise, it is  \emph{value indefinite} (under $v$). Similarly, we call $O$  \emph{value definite} (under $v$) if every observable $P \in O$ is value definite.
\end{definition}

%
%

\section{Theoretical basis}
\label{val_in}

In this section  we   present the main theoretical basis of the QRNG.

\subsection{Localising value indefiniteness}

Consider the following  Kochen-Specker assumptions:

\begin{itemize}
\item {\bf Admissibility:} Let $O$ be a set of one-dimensional projection observables on $\mathbb{C}^{n}$ and let $v:O\rightarrow \{0,1\}$ be a value assignment function. Then $v$ is  \emph{admissible}\footnote{In agreement with quantum mechanics predictions.} if for  every context $C$ of $O$, we have that $\sum_{P\in C}v(P) = 1$, i.e.~only one projection observable in a context can be assigned the value $1$.

\item {\bf Non-contextuality of definite values}:  The outcome obtained by
measuring a value definite observable
(a pre-existing physical property) is \emph{non-contextual}, i.e.\ it does not
depend on other compatible
 observables which may be measured alongside it.
\end{itemize}

The fundamental result is:
\begin{theorem}[Kochen-Specker~\cite{Kochen:2017aa}]
\label{KS}
Let $n \geq 3$. Then there exists a (finite) set of one-dimensional projection observables $O$ on the Hilbert space $\mathbb{C}^{n}$ such that there is no value assignment function $v$ satisfying the following 
three conditions: i) $O$ is value definite under $v$, ii)  $v$ is admissible, iii) $v$ is non-contextual.
\end{theorem}

Kochen-Specker Theorem shows that, in agreement with quantum mechanics, not every observable can be both non-contextual and value definite, but it does not describe the extent of this incompatibility. In fact,  it has been shown that for any sets of observables there exists an admissible assignment function under which the set of observables is value definite and at least one observable is non-contextual. That is, the incompatibility between the Kochen-Specker assumptions is not maximal, hence not all observables need to be value indefinite.

Why are value indefinite observables important? One reason is that {\it measuring one such observable may produce a random outcome}. But, to measure a value indefinite observable we have to ``effectively find" one, not just know that such an observable exists as Kochen-Specker Theorem assures. Essentially, to answer the above question in the   affirmative, we need a constructive form of the Kochen-Specker Theorem allowing 
to localise a value indefinite observable.  Motivated by  Einstein, Podolsky and Rosen  definition of {\em physical reality}~\cite[p.~777]{epr}:
 
 \begin{quote}
	 If, without in any way disturbing a system, we can predict with certainty the value of a physical quantity, then there exists a \emph{definite value} prior to observation corresponding to this physical quantity.
\end{quote}

\noindent we make the following assumption:

\begin{itemize}
\item {\bf Eigenstate principle}:  If a quantum system is prepared in the
state $\ket{\psi}$, then the projection observable
$P_\psi$ is value definite.
\end{itemize}

In detail,  if a quantum system is prepared in an arbitrary state $\ket{\psi}\in\mathbb{C}^{n}$, then the measurement of the observable $P_{\psi}$ should yield the outcome $1$, hence, if $P_{\psi}\in O$, then $v(P_{\psi})=1.$


\begin{theorem}[Localised Kochen-Specker~\cite{abbott2012strongrandomness,vi-aeverywhere-2014,2015-AnalyticKS}]
\label{EffecKS}
Assume a quantum system prepared in the state
$\ket{\psi}$ in a dimension $n\ge 3$ Hilbert space ${\mathbf C}^n$, and let $\ket{\phi}$
be any state neither orthogonal nor parallel to $\ket{\psi}$ {\rm (}$0<\abs{\bra{\psi}\ket{\phi}}<1${\rm )}. If the following three conditions are satisfied: i) admissibility, ii) non-contextuality and iii) 
eigenstate principle, 
then the projection observable $P_\psi$ is  value
indefinite.
\end{theorem}

From Theorem~\ref{EffecKS} we deduce that, given a system prepared in state $\ket{\psi}$, a one-dimensional projection observable can only be value definite if it is an eigenstate of that observable.
Furthermore, for any diagonalisable observable $O$ with spectral decomposition $O=\sum_{i=1}^n\lambda_{i}P_{\lambda_{i}}$, where $\lambda_{i}$ denotes each distinct eigenvalue with corresponding eigenstate $\ket{\lambda_{i}}$,  $O$ has a predetermined measurement outcome if and only if each projector in its spectral decomposition has a predetermined measurement outcome. Thus, we can generalise our previous result to the outcome of the measurement of an observable with non-degenerate spectra. 
Such generalisation is of particular importance for applying this result to elements of physical reality where a measurement is assumed to yield a meaningful result that describes a physical attribute; thus, utilising the value assignment function to represent the realisation of a given state whenever the corresponding observable is value definite. The latter can be observed as follows.

Let $C = \{P_{1},\ldots,P_{n}\}$ be a context of projection observables and let $v$ be a value assignment function such that $v(P_{1})=1$ under $C$.  Since a context is a maximal set of compatible projection observables it follows that, if any pair $(P_{1},P_{i})$ is measured, then the system will collapse into the eigenstate $\ket{\phi}$ of the projection observable $P_{1}$ with eigenvalue $1$. It follows that, as all observables in $C$ are physically co-measurable and $\sum_{j=1}^nP_{j}=1$, we deduce that $\ket{\phi}$ is an eigenstate of $P_{i}$ with corresponding eigenvalue $0$; that is, $v(P_{i})=0$. Similarly, if $v(P_{i}) = 0$ for $i\neq 1$, then $v(P_{1})=1$. Hence, the admissibility of $v$ serves as a generalisation of the sum rule that corresponds to the physical interpretation of the measurement process.

Finally, we can answer the question `how ``large" or ``typical" is the set of value indefinite observables?'

\begin{theorem}
	[\cite{vi-aeverywhere-2014}]  The set of value indefinite observables has constructive
Lebesgue  measure one, that is, almost all observables are value indefinite.
\end{theorem}

Theorem~\ref{EffecKS} paved the way to construct a class of QRNGs based on measuring value indefinite observables. How ``good" is such a QRNG?  The answer will  use the following

\begin{itemize}
	\item  {\bf epr  principle}: If a repetition of measurements of an
observable generates a computable sequence, then  these observables
are value definite.
\end{itemize}

Assume the  Eigenstate and epr    principles. An infinite repetition of the
experiment 
measuring  a quantum value indefinite observable   always generates an
incomputable   infinite sequence $x_1x_2\dots$. In fact, a stronger result is true  as we will show in Section~\ref{3biimm}. Informally, a sequence $ \mathbf{x}$ is bi-immune if  no algorithm can generate
infinitely many correct values of its elements (pairs, $(i, x_i)$).
 The formal definition is as follows.
A sequence $ \mathbf{x}\in A_{b}^{\omega}$  ($b\ge 2$)
 is  {\it bi-immune} if there is no
partially computable function $\varphi$ from $\bbbn $ to $A_{b}$ having an
infinite domain $\dom(\varphi)$ with the property that $\varphi(i)= x_i$ for
all $i\in \dom(\varphi)$~\cite{bienvenu2013}). 
In the binary case we have:

\begin{theorem}[\cite{abbott2012strongrandomness}]
\label{biimm_2}
Assume the  Eigenstate and epr    principles. An infinite repetition of the
experiment 
measuring  a quantum value indefinite observable  in $\mathbb{C}^{2}$ always generates a
 bi-immune   sequence $\x\in A_2^{\omega}$. 
\end{theorem}

\subsection{Value definiteness and unpredictability}

Since probability spaces lie at the core of quantum mechanics, we can describe quantum behaviour in different contexts by utilising the probabilistic framework that the theoretical notion of the wave function characterisation provides; here, physical attributes correspond to projection operators and their corresponding eigenvalues.
However, the use of the \emph{eigenstate assumption} is restricted to contexts that contain the observable $P_{\psi}$, where $\ket{\psi}$ is the state in which the system was prepared.  For this reason, formalising the notion of predictability with respect to the value that corresponds to a given observable is required.

Consider a system that continuously repeats the process of state preparation and measurement, as in \cite{abbott2012strongrandomness}. Let $\textbf{x} = x_1x_2\ldots$ denote the infinite sequence produced by concatenating the outputs of the measurement performed at each iteration. Let $\mathcal{O,C}$ be a fixed set of observables and contexts, respectively, with $O_i,C_i$ denoting the observable and the corresponding context for the $i$-th measurement. 
We say that a measurement outcome is predictable if there exists a \emph{computable function} $f: \mathbb{N}\times\mathcal{O}\times\mathcal{C}\rightarrow\{0,1\}$ such that, for every iteration $i$ we have that $f(i,O_i,C_i)=x_i$.
Note that if every value of a sequence of measurement results is predictable, then the computability of $f$ ensures that there is some function that outputs the  values $x_i$ of $\textbf{x}$ corresponding to each iteration. However, an  incomputable $f$  provides no way of obtaining each term of the sequence and therefore offers no method of prediction \cite{Sipser:1996:ITC:524279}. Finally, following~\cite{abbott2012strongrandomness}, if such function exists, we assume there is a definite value associated with the sequence of observables used for computing each term of the function output; that is $f(i,O_i,C_i)=v_i(O_i,C_i)$. 

Theorem~\ref{biimm_2}  proves this form of unpredictability,  but leaves the possibility of finitely many exceptions. An even stronger result,  which removes this possibility,  was obtained by
using a non-probabilistic  model for unpredictability~\cite{acs-2015-info6040773,DBLP:conf/birthday/AbbottCS15}.  To this aim we consider an \emph{experiment} $E$ producing a single bit $x\in\{0,1\}$; with a particular  trial of $E$ we associate the parameter $\lambda$  (the state of the universe) which fully describes the trial;  $\lambda$ can be viewed as a resource from which one can extract finite information in order   to predict the outcome of the experiment $E$. The trials of $E$ generate a succession of events of the form ``$E$ is prepared, performed, the result recorded, $E$ is reset'', iterated finitely many times in an algorithmic fashion. 

 An \emph{extractor} is a physical device selecting a finite amount of information from  $\lambda$ without altering the experiment $E$; 
it  produces  a finite string of bits $\langle \lambda \rangle $.
A \emph{predictor} for $E$ is an algorithm   $P_E$
which \emph{halts} on every input and \emph{produces} {\bf 0} or {\bf 1}  or {\bf prediction withheld}.The predictor  $P_E$ can use as input the information $\langle \lambda \rangle $, but
must be \emph{passive,} that is, it must not disturb or interact with $E$ in any way.

A predictor $P_E$ provides a \emph{correct prediction} using the extractor $\langle \, \rangle$ for an instantiation of $E$ with parameter $\lambda$ if, when taking as input $\langle \lambda \rangle$,
it outputs 0 or 1 (i.e.\ it does not refrain from making a prediction) and the output is equal to $x$, the result of the experiment.
Fix an extractor $\langle \, \rangle$; the predictor $P_E$ is {\em $k,\langle \, \rangle$-correct} if there exists an $n\ge k$ such that when $E$ is repeated $n$ times with associated parameters $\lambda_1 ,\dots, \lambda_n$ producing the outputs $x_1,x_2,\dots ,x_n$, $P_E$ outputs the sequence $P_E(\langle\lambda_1\rangle), P_E(\langle\lambda_2\rangle),\dots ,P_E(\langle\lambda_n\rangle)$ with the following two properties:  (i) no prediction in the sequence is incorrect, and (ii) in the sequence there are $k$ correct predictions.
The confidence we have in a $k,\langle \, \rangle$-correct predictor increases as $k\to\infty$.
If $P_E$ is $k,\langle \, \rangle$-correct for all $k$, then $P_E$ never makes an incorrect prediction and the number of correct predictions can be made arbitrarily large by repeating $E$ enough times.

If $P_E$ is not $k,\langle \, \rangle$-correct  for all $k$, then we cannot exclude the possibility that any correct prediction $P_E$  makes is simply due to chance. Hence, we say that
the outcome $x$ of a single trial of the experiment $E$ performed with parameter $\lambda$ is \emph{predictable} (with certainty) if there exist an extractor $\langle \, \rangle$ and a predictor $P_E$ which is $k,\langle \, \rangle$-correct for all $k$, and $P_E(\langle\lambda\rangle)=x$.

Consider an experiment $E$ performed in dimension $n\ge 3$ Hilbert space in which a quantum system is prepared in a state $\ket{\psi}$ and a value indefinite observable $P_\phi$ is measured producing a single bit $x$.


\begin{theorem}[\cite{acs-2015-info6040773,DBLP:conf/birthday/AbbottCS15}]  
\label{2unpredict}
Assume the epr and Eigenstate  principles.
Let $\x$ be an infinite sequence obtained by measuring  a quantum value indefinite observable  in $\mathbb{C}^{2}$  in an infinite repetition of the experiment $E$. Then no single bit $x_i$ can be predicted.

\end{theorem}

\section{A QNRG based on localised value indefiniteness}
\label{qrng1}

A blueprint for a QRNG based on Theorem~\ref{EffecKS} was proposed in~\cite{abbott2012strongrandomness} using a 
generalised beam splitter and a physical realisation
with superconducting transmon qutrits was given in~\cite{PhysRevLett.119.240501}. 
As Theorems~\ref{KS} and~\ref{EffecKS} are true {\em only} in  Hilbert spaces of dimension $n\ge 3$, any QRNG based on them produces sequences over alphabets with at least three elements. As a consequence, a QRNG using the classical beam splitter is not certified by these theorems.

{\it The QRNG operates in a  succession of events of the form ``preparation,  measurement,
 reset'', iterated indefinitely many times in an algorithmic fashion},\cite{abbott2012strongrandomness}. Let $\textbf{x} = x_1x_2\ldots$ denote the infinite sequence produced by concatenating the consecutive outputs of infinitely many events as described above.

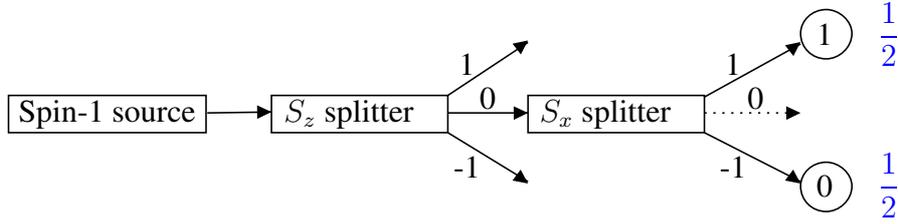
\begin{figure}[ht]
\centering
\tikzset{every picture/.style={line width=0.60pt}} 
\begin{tikzpicture}[x=0.6pt,y=0.6pt,yscale=-1,xscale=1]
\draw   (26,129) -- (149.2,129) -- (149.2,152.4) -- (26,152.4) -- cycle ;
\draw    (150,140) -- (190,139.42) ;
\draw [shift={(190,139.4)}, rotate = 539.51] [fill={rgb, 255:red, 0; green, 0; blue, 0 }  ][line width=0.75]  [draw opacity=0] (8.93,-4.29) -- (0,0) -- (8.93,4.29) -- cycle    ;
\draw   (190,129) -- (300,129) -- (300,152.4) -- (190,152.4) -- cycle ;
\draw    (300,140) -- (350,140) ;
\draw [shift={(350,140)}, rotate = 539.51] [fill={rgb, 255:red, 0; green, 0; blue, 0 }  ][line width=0.75]  [draw opacity=0] (8.93,-4.29) -- (0,0) -- (8.93,4.29) -- cycle    ;
\draw   (350,129) -- (460,129) -- (460,152.4) -- (350,152.4) -- cycle ;
\draw [dash pattern={on 0.84pt off 2.51pt}](460,140) -- (520,140) ;
\draw [shift={(520,140)}, rotate = 539.51] [fill={rgb, 255:red, 0; green, 0; blue, 0 }  ][line width=0.75]  [draw opacity=0] (8.93,-4.29) -- (0,0) -- (8.93,4.29) -- cycle    ;
\draw    (300,152.4) -- (350,183.21) ;
\draw [shift={(350,184.4)}, rotate = 216.66] [fill={rgb, 255:red, 0; green, 0; blue, 0 }  ][line width=0.75]  [draw opacity=0] (8.93,-4.29) -- (0,0) -- (8.93,4.29) -- cycle    ;
\draw    (300,129) -- (350,95.62) ;
\draw [shift={(350,94.4)}, rotate = 502.44] [fill={rgb, 255:red, 0; green, 0; blue, 0 }  ][line width=0.75]  [draw opacity=0] (8.93,-4.29) -- (0,0) -- (8.93,4.29) -- cycle    ;
\draw    (460,130) -- (520,96.62) ;
\draw [shift={(520,95.4)}, rotate = 502.44] [fill={rgb, 255:red, 0; green, 0; blue, 0 }  ][line width=0.75]  [draw opacity=0] (8.93,-4.29) -- (0,0) -- (8.93,4.29) -- cycle    ;
\draw    (460,152.4) -- (520,183.21) ;
\draw [shift={(520,184.4)}, rotate = 216.66] [fill={rgb, 255:red, 0; green, 0; blue, 0 }  ][line width=0.75]  [draw opacity=0] (8.93,-4.29) -- (0,0) -- (8.93,4.29) -- cycle    ;
\draw   (520,90) .. controls (520,82) and (527,75) .. (536,75) .. controls (544,75) and (552,82) .. (552,90) .. controls (552,99) and (544,106) .. (536,106) .. controls (527,106) and (520,99) .. (520,90) -- cycle ;
\draw   (520,186) .. controls (520,178) and (527,171) .. (536,171) .. controls (544,171) and (552,178) .. (552,186) .. controls (552,195) and (544,202) .. (536,202) .. controls (527,202) and (520,195) .. (520,186) -- cycle ;
\draw (87.6,140.7) node  [align=left] {Spin-1 source};
\draw (240.6,140.7) node  [align=left] {$\displaystyle S_{z}$ splitter };
\draw (400.6,140.7) node  [align=left] {$\displaystyle S_{x}$ splitter };
\draw (312,109) node  [align=left] {1};
\draw (312,174) node  [align=left] {\mbox{-}1};
\draw (325,131) node  [align=left] {0};
\draw (535,90.7) node  [align=left] {1};
\draw (535,185.7) node  [align=left] {0};
\draw (478,109) node  [align=left] {1};
\draw (478,174) node  [align=left] {\mbox{-}1};
\draw (492,131) node  [align=left] {0};
\draw (575,90.7) node  [align=left] {$\displaystyle\color{blue}{\frac{1}{2}}$};
\draw (575,185.7) node  [align=left] {$\displaystyle\color{blue}{\frac{1}{2}}$};
\end{tikzpicture}
\caption{QRNG setup proposed in \cite{abbott2012strongrandomness}; the values $\frac{1}{2},\frac{1}{2}$ (in blue) correspond to the outcome probabilities}
\label{fig1}
\end{figure}

From Theorem~\ref{EffecKS}  a system prepared on an arbitrary state $\ket{\psi}$ must have a definite value associated to the operator $P_{\psi}$. Hence, for spin-1 particles prepared in the state $S_{z} = 0$,  this operator is value definite.
As the possible outcomes of an observable $O$ correspond to the eigenvalues $o$ of the projectors that describe the spectral decomposition $O = \sum_o oP_{o}$,
we deduce  that the state $\ket{S_{z} = 0}$ is an eigenstate of the projector $S_{x} = 0$, i.e.~$\ket{0}\bra{0}$, with eigenvalue $0$; so, the probability of obtaining this outcome is 0.
For this reason, $S_{x} = \pm 1$ are the only results we need to consider for now. Furthermore, we have that $\braket{S_{z}}{S_{x}} = \braket{0}{\pm 1} = \frac{1}{\sqrt{2}}$; so, by the previous results, it is not possible to assign a definite value to  $S_{x} = \pm 1$.

To date, \emph{this QRNG is the only example of a random generator provably better than any PRNG.} 

An experimental study~\cite{Abbott_2019}  of the realisation~\cite{PhysRevLett.119.240501} of this QRNG has used various tests to compare it with arguably the best  PRNGs. While the analysis failed to observe a strong advantage of the quantum random sequences due to incomputability, the results are  informative: some of the test results are ambiguous and require further study, others highlight difficulties that can guide the development of future tests of algorithmic randomness and incomputability, and, more importantly,  ideas for improvement of the design of QRNG based  on Theorem~\ref{EffecKS} have emerged. One such idea, developed in the following section, is to eliminate the problematic branch  $S_x=0$  in Figure~\ref{fig1} which has probability zero. Why problematic? In standard measure-theoretic formulation of probability \cite{feller1} it is possible for a non-empty event to have probability zero, hence, events of probability zero are not necessarily impossible.

\section{A new QRNG based on localised value indefiniteness}
\label{qrng2}

To address the above problem  we  propose  a new QRNG setup, based on the blueprint, with a
 different state preparation, see Figure~\ref{fig2}. 

\begin{figure}
\centering
\tikzset{every picture/.style={line width=0.60pt}} 
\begin{tikzpicture}[x=0.6pt,y=0.6pt,yscale=-1,xscale=1]
\draw   (26,129) -- (149.2,129) -- (149.2,152.4) -- (26,152.4) -- cycle ;
\draw    (150,140) -- (190,139.42) ;
\draw [shift={(190,139.4)}, rotate = 539.51] [fill={rgb, 255:red, 0; green, 0; blue, 0 }  ][line width=0.75]  [draw opacity=0] (8.93,-4.29) -- (0,0) -- (8.93,4.29) -- cycle    ;
\draw   (190,129) -- (300,129) -- (300,152.4) -- (190,152.4) -- cycle ;
\draw    (300,140) -- (350,140) ;
\draw [shift={(350,140)}, rotate = 539.51] [fill={rgb, 255:red, 0; green, 0; blue, 0 }  ][line width=0.75]  [draw opacity=0] (8.93,-4.29) -- (0,0) -- (8.93,4.29) -- cycle    ;
\draw   (350,129) -- (460,129) -- (460,152.4) -- (350,152.4) -- cycle ;
\draw (460,140) -- (520,140) ;
\draw [shift={(520,140)}, rotate = 539.51] [fill={rgb, 255:red, 0; green, 0; blue, 0 }  ][line width=0.75]  [draw opacity=0] (8.93,-4.29) -- (0,0) -- (8.93,4.29) -- cycle    ;
\draw    (460,130) -- (520,96.62) ;
\draw [shift={(520,95.4)}, rotate = 502.44] [fill={rgb, 255:red, 0; green, 0; blue, 0 }  ][line width=0.75]  [draw opacity=0] (8.93,-4.29) -- (0,0) -- (8.93,4.29) -- cycle    ;
\draw    (460,152.4) -- (520,183.21) ;
\draw [shift={(520,184.4)}, rotate = 216.66] [fill={rgb, 255:red, 0; green, 0; blue, 0 }  ][line width=0.75]  [draw opacity=0] (8.93,-4.29) -- (0,0) -- (8.93,4.29) -- cycle    ;
\draw   (520,90) .. controls (520,82) and (527,75) .. (536,75) .. controls (544,75) and (552,82) .. (552,90) .. controls (552,99) and (544,106) .. (536,106) .. controls (527,106) and (520,99) .. (520,90) -- cycle ;
\draw   (520,138) .. controls (520,130) and (527,123) .. (536,123) .. controls (544,123) and (552,130) .. (552,138) .. controls (552,147) and (544,154) .. (536,154) .. controls (527,154) and (520,147) .. (520,138) -- cycle ;
\draw   (520,186) .. controls (520,178) and (527,171) .. (536,171) .. controls (544,171) and (552,178) .. (552,186) .. controls (552,195) and (544,202) .. (536,202) .. controls (527,202) and (520,195) .. (520,186) -- cycle ;
\draw (87.6,140.7) node  [align=left] {Spin-1 source};
\draw (240.6,140.7) node  [align=left] {$\displaystyle S_{z}$ splitter };
\draw (400.6,140.7) node  [align=left] {$\displaystyle S_{x}$ splitter };
\draw (325,121) node  [align=left] {$\ket{\psi}$};
\draw (535,90.7) node  [align=left] {0};
\draw (535,137.7) node  [align=left] {1};
\draw (535,185.7) node  [align=left] {2};
\draw (478,109) node  [align=left] {1};
\draw (478,174) node  [align=left] {\mbox{-}1};
\draw (492,131) node  [align=left] {0};
\draw (575,90.7) node  [align=left] {$\displaystyle\color{blue}{\frac{1}{4}}$};
\draw (575,137.7) node  [align=left] {$\displaystyle\color{blue}{\frac{1}{2}}$};
\draw (575,185.7) node  [align=left] {$\displaystyle\color{blue}{\frac{1}{4}}$};
\end{tikzpicture}
\caption{Blueprint for a new QRNG; the values $\frac{1}{4}, \frac{1}{2},\frac{1}{4}$ (in blue)  correspond to the outcome probabilities of setups prepared in the  state $\ket{\psi} = \ket{\pm 1}$}
\label{fig2}
\end{figure}
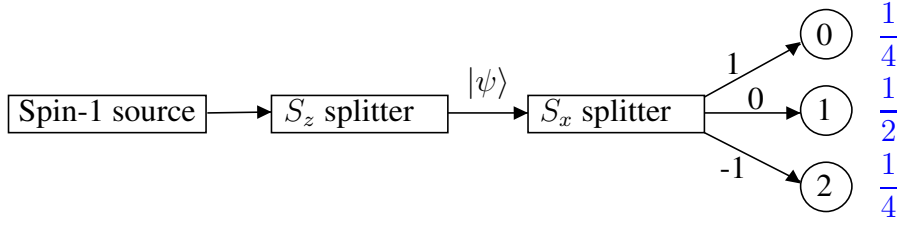

\subsection{A generalised spin observable}

The property  \emph{spin} ($\mathbf{S}$) is the intrinsic angular momentum characteristic of elementary particles. By deriving the spin state operator $S_x$ we can control the effect of the preparation state $\ket{S_z}$ on the outcome probabilities. We refer to the eigenvalue $s$ of $\mathbf{S}^2$ as the spin (quantum) number~\cite{Merzbacher:1998aa,svozil-pu2018}.
For a spin-1 particle, the eigenvalues of $S_{z}$ are $1,0,-1$, thus introducing an orthonormal Cartesian standard basis $\{\ket{1} , \ket{0}, \ket{-1}\}$ defined by $S_z\ket{m} = \hbar m\ket{m}$ it follows that 
$$ S_z = \hbar 
\begin{pmatrix}
1 & 0 & 0\\
0 & 0 & 0 \\
0 & 0 &  -1
\end{pmatrix}
\rb.$$
From $S_\pm\ket{m}=\sqrt{s(s+1)-m(m\pm 1)}\ket{m\pm 1} $ we obtain the raising and lowering operators for $s=1$

 $$S_+\ket{m} = \hbar\sqrt{2-m(m+1)}\ket{m+1} ,$$ $$S_-\ket{m} = \hbar\sqrt{2-m(m-1)}\ket{m-1}.$$ Consequently, we have 
 
 $$ S_+ = 
 \begin{pmatrix}
 \bra{1}S_+\ket{1}&  \bra{1}S_+\ket{0}& \bra{1}S_+\ket{-1} \\
 \bra{0}S_+\ket{1} & \bra{0}S_+\ket{0} &  \bra{0}S_+\ket{-1}\\
  \bra{-1}S_+\ket{1}& \bra{-1}S_+\ket{0} &  \bra{-1}S_+\ket{-1}
 \end{pmatrix}
 = \sqrt{2}\hbar
 \begin{pmatrix} 0&1&0\\0&0&1\\0&0&0\end{pmatrix}\rb,$$
  
 $$ S_- = 
 \begin{pmatrix}
 \bra{1}S_-\ket{1}&  \bra{1}S_-\ket{0}& \bra{1}S_-\ket{-1} \\
 \bra{0}S_-\ket{1} & \bra{0}S_-\ket{0} &  \bra{0}S_-\ket{-1}\\
  \bra{-1}S_-\ket{1}& \bra{-1}S_-\ket{0} &  \bra{-1}S_-\ket{-1}
 \end{pmatrix}
 = \sqrt{2}\hbar
 \begin{pmatrix}0&0&0\\1&0&0\\0&1&0 \end{pmatrix}\rb.$$ 
 
 Furthermore, since $S_\pm = S_x \pm iSy$, we get $S_x = \frac{1}{2}(S_+ + S_-)$ and $S_y = \frac{1}{2i} (S_- - S_+)$, it follows that
 $$S_x = \frac{1}{\sqrt{2}}\hbar
 \begin{pmatrix}
0&1&0\\
1&0&1\\
0&1&0 
 \end{pmatrix},\text{ }
S_y = \frac{1}{\sqrt{2}}\hbar
 \begin{pmatrix}
0&-i&0\\
i&0&-i\\
0&i&0 
 \end{pmatrix}\rb.
 $$ Thus, the generalised Pauli matrices for a spin-1 particle are given by $\textbf{S} = (S_x,S_y,S_z) = \hbar\pmb{\sigma}$:
 $$ 
 \sigma_x = \frac{1}{\sqrt{2}}
 \begin{pmatrix}
0&1&0\\
1&0&1\\
0&1&0 
 \end{pmatrix}\rb,
\text{ }
 \sigma_y = \frac{1}{\sqrt{2}}
  \begin{pmatrix}
0&-i&0\\
i&0&-i\\
0&i&0 
 \end{pmatrix}\rb,
\text{ }
 \sigma_z = 
 \begin{pmatrix}
1 & 0 & 0\\
0 & 0 & 0 \\
0 & 0 &  -1
\end{pmatrix}\rb.
 $$ 
 We can now consider the description of spin states that point in arbitrary directions specified by the unit vector $\textbf{u} = (u_x, u_y, u_z) = (\sin\theta\cos\phi , \sin\theta\sin\phi, \cos\theta)$, where $\theta ,\phi$ are the polar and azimuthal angles; we then define the spin observable operator $\textbf{S}$ as a triplet of operators $\textbf{S} = (S_x,S_y,S_z) = \hbar\pmb{\sigma}$. Then, by adopting units in which $\hbar$ is numerically equal to unity, in order to reduce the amount of numerical clutter, we obtain the generalised spin observable operator that describes the measurement context:
 $$S(\theta,\phi) = \textbf{u}\cdot\textbf{S} =
 \begin{pmatrix}
 u_z & \frac{ u_x - iu_y}{\sqrt{2}} & 0\\   
\frac{ u_x + iu_y}{\sqrt{2}} & 0 & \frac{u_x - iu_y}{\sqrt{2}}\\
 0 & \frac{u_x + iu_y}{\sqrt{2}} & - u_z
 \end{pmatrix}\rb,$$ that is,
$$ S(\theta, \phi) = 
\begin{pmatrix}
\cos(\theta) & \frac{e^{-i\phi}\sin(\theta)}{\sqrt{2}} & 0\\
\frac{e^{i\phi}\sin(\theta)}{\sqrt{2}} & 0 & \frac{e^{-i\phi}\sin(\theta)}{\sqrt{2}}\\
0 & \frac{e^{i\phi}\sin(\theta)}{\sqrt{2}}&  -\cos(\theta)
\end{pmatrix}\rb.$$ 

Note that $S_z$ is given by $S(0,0)$ and $S_x$ by $S(\frac{\pi}{2},0).$

\subsection{State preparation and outcome probabilities}

By considering the orthonormal Cartesian standard basis $\ket{1} = (1,0,0), \ket{0} = (0,1,0)$ and $\ket{-1} = (0,0,1)$
we can obtain the eigenvalues $\{-1,0,1\}$ of $S_x$ by solving the equation
%
 
 $$ det(S_x - I\lambda )  = 
\begin{vmatrix}
-\lambda & \frac{1}{\sqrt{2}} & 0\\
\frac{1}{\sqrt{2}} & -\lambda & \frac{1}{\sqrt{2}} \\
0 & \frac{1}{\sqrt{2}} & -\lambda
\end{vmatrix}
 = 0,$$
 
\noindent that is, $-\lambda (\lambda^2 - \frac{1}{2}) + \frac{1}{2}\lambda = 0.$ Consequently
 we have:
 \begin{itemize}
\item[] $S_x\ket{S_x:1} =\ket{S_x:1} \implies  \ket{S_x:+1} = \frac{1}{2} (1,\sqrt{2},1)   ,$
\item[] $S_x\ket{S_x:0} =0 \implies  \ket{S_x:0} = \frac{1}{\sqrt{2}} (1,0,-1)   ,$
\item[] $S_x\ket{S_x:-1} =-\ket{S_x:-1} \implies  \ket{S_x:+1} = \frac{1}{2} (1,-\sqrt{2},1)   .$
 \end{itemize}
 
 We are now able to form the unitary matrix $U_{x}$ corresponding to the spin state operator $S_x$
 $$ U_x = \frac{1}{2}
 \begin{pmatrix}
 1 & \sqrt{2} & 1\\
 \sqrt{2} & 0 & -\sqrt{2}\\
 1 & -\sqrt{2} & 1
 \end{pmatrix}\rb,$$
  
\noindent The fact that  $U_x$ can be decomposed into two-dimensional transformations \cite{clements_optimal_2016} enables the physical realisation of the unitary operator by a lossless beam splitter \cite{PhysRevLett.73.58,yurke-86}  leading to  the implementation of a QRNG, as in~\cite{PhysRevLett.119.240501}, with the new outcome probabilities.
For simplicity we adopt the following convention:
 \begin{itemize}
\item[] $ \ket{1_x} = \ket{S_x:+1} = \frac{1}{2}\ket{1} + \frac{1}{\sqrt{2}}\ket{0} + \frac{1}{2}\ket{-1}\rb,$
\item[] $ \ket{0_x} = \ket{S_x:0} = \frac{1}{\sqrt{2}}\ket{1} - \frac{1}{\sqrt{2}}\ket{-1}\rb,$
\item [] $ \ket{-1_x}  =  \ket{S_x:+1} =\frac{1}{2}\ket{1} - \frac{1}{\sqrt{2}}\ket{0} + \frac{1}{2}\ket{-1}\rb.$
 \end{itemize}
 
Consider the probability distribution $\frac{1}{4},\frac{1}{2},\frac{1}{4}\rb.$
We can identify a possible corresponding state preparation $\ket{\psi}$ by solving the following system of equations:

\begin{itemize}
\item[] $|\frac{1}{2}x + \frac{1}{\sqrt{2}}y + \frac{1}{2}z |= \frac{1}{2}\rb,$
\item[] $|\frac{1}{\sqrt{2}}x - \frac{1}{\sqrt{2}}z|= \frac{1}{\sqrt{2}}\rb,$
\item[] $|\frac{1}{2}x - \frac{1}{\sqrt{2}}y + \frac{1}{2}z |= \frac{1}{2}\rb,$
\end{itemize}

\noindent where $x=\braket{1}{\psi}, y=\braket{0}{\psi}, z=\braket{-1}{\psi}$.
Setting $y=0,	z=1-x$ satisfies such constrains and  provides $\ket{+1},\ket{-1}$ and $\frac{\ket{+}-\ket{-}}{\sqrt{2}}$\rb, where $\ket{\pm} = \frac{1}{\sqrt{2}}\ket{0} \pm \frac{1}{\sqrt{2}}\ket{1}$, as preparation state candidates. Since $\ket{+1}$ and $\ket{-1}$ are eigenstates of $S_z$ they represent a natural choice for our QRNG construction.
We ensure the validity of these states by first noting that

\begin{itemize}
 \item[] $\braket{1_x}{1} = \frac{1}{2}\rb, \;\; \braket{1_x}{-1} = \frac{1}{2}\rb,$
 \item[] $\braket{0_x}{1} = \frac{1}{\sqrt{2}}\rb,\;\; \braket{0_x}{-1} = \frac{-1}{\sqrt{2}}\rb,$
 \item[] $\braket{-1_x}{1} = \frac{1}{2}\rb,\;\; \braket{-1_x}{-1} = \frac{1}{2}\rb.$
\end{itemize}

Then, for $\ket{\psi}=\ket{\pm 1}$ we have
$$\braket{1_x}{\psi} = \frac{1}{2},\; \braket{0_x}{\psi} = \pm\frac{1}{\sqrt{2}},\; \braket{-1_x}{\psi} = \frac{1}{2}\rb.$$

Furthermore, we have 
\begin{itemize}
\item[] $\braket{1_x}{+} = \frac{1}{\sqrt{2}}(\braket{1_x}{0} + \braket{1_x}{1}) =\frac{1}{\sqrt{2}}(\frac{1}{\sqrt{2}} + \frac{1}{2}) = \frac{1}{2}+\frac{1}{2\sqrt{2}}$\rb,
\item[] $\braket{1_x}{-} = \frac{1}{\sqrt{2}}(\braket{1_x}{0} - \braket{1_x}{1}) =\frac{1}{\sqrt{2}}(\frac{1}{\sqrt{2}} - \frac{1}{2}) = \frac{1}{2}- \frac{1}{2\sqrt{2}};$\\

\item[] $\braket{0_x}{+}=\frac{1}{\sqrt{2}}(\braket{0_x}{0}+\braket{0_x}{1})=\frac{1}{2}\rb,$
\item[] $\braket{0_x}{-} =\frac{1}{\sqrt{2}}(\braket{0_x}{0} - \braket{0_x}{1}) = -\frac{1}{2};$\\

\item[] $\braket{-1_x}{+} = \frac{1}{\sqrt{2}}(\braket{-1_x}{0} + \braket{-1_x}{1}) =\frac{1}{\sqrt{2}}(-\frac{1}{\sqrt{2}} + \frac{1}{2}) = -\frac{1}{2}+\frac{1}{2\sqrt{2}}\rb,$
\item[] $\braket{-1_x}{-} = \frac{1}{\sqrt{2}}(\braket{-1_x}{0} - \braket{1_x}{1}) =\frac{1}{\sqrt{2}}(-\frac{1}{\sqrt{2}} - \frac{1}{2}) = -\frac{1}{2}-\frac{1}{2\sqrt{2}}\rb.$
\end{itemize}

Thus, for $\ket{\psi}={\frac{\ket{+}-\ket{-}}{\sqrt{2}}}$\rb, we have:
 \begin{itemize}
 \item[] $\braket{1_x}{\psi} = \frac{1}{\sqrt{2}}(\braket{1_x}{+} - \braket{1_x}{-}) = \frac{1}{\sqrt{2}}( \frac{1}{2}+\frac{1}{2\sqrt{2}} - \frac{1}{2}+ \frac{1}{2\sqrt{2}}) = \frac{1}{2}$\rb,
 \item[] $\braket{0_x}{\psi} = \frac{1}{\sqrt{2}}(\braket{0_x}{+} - \braket{0_x}{-}) = \frac{1}{\sqrt{2}}( \frac{1}{2} + \frac{1}{2}) = \frac{1}{\sqrt{2}}$\rb,
 \item[] $\braket{-1_x}{\psi} = \frac{1}{\sqrt{2}}(\braket{-1_x}{+} - \braket{-1_x}{-}) = \frac{1}{\sqrt{2}}( -\frac{1}{2}+\frac{1}{2\sqrt{2}} + \frac{1}{2}+ \frac{1}{2\sqrt{2}}) = \frac{1}{2}$\rb.
 \end{itemize}
 
Hence, by  the third postulate of quantum mechanics, we obtain the following probabilities for $\ket{\psi}\in\{\ket{+ 1},\ket{-1},\frac{\ket{+}-\ket{-}}{\sqrt{2}}\}$:
 \begin{itemize}
 \item[] $p(S_{x,1}) = |\braket{1_x}{\psi}|^2 = \frac{1}{4}\rb,$
 \item[] $p(S_{x,0}) = |\braket{0_x}{\psi}|^2 = \frac{1}{2}\rb,$
 \item[] $p(S_{x,-1}) = |\braket{-1_x}{\psi}|^2 = \frac{1}{4}\rb.$
 \end{itemize}

 From these results it is clear that the preparation states $\ket{+1}$, $\ket{-1}$ and $\frac{\ket{+}-\ket{-}}{\sqrt{2}}$ for obtaining the outcome probabilities $\frac{1}{4}\rb ,\frac{1}{2}\rb ,\frac{1}{4}$ satisfy the requirements of Theorem~\ref{EffecKS}. Furthermore, as only the preparation state $\ket{S_z}$ is modified, the unitary matrix $U_x$ remains unaltered, thus enabling the physical realisation of this QRNG.
\medskip

{\it In what follows by QRNG will mean the QRNG constructed in this section.}

\subsection{Ternary 3-bi-immunity}
\label{3}

In this section we study the main properties of quantum random sequences produced by the proposed QRNG:  3-bi-immunity, unpredictability and Borel normality.

\subsection{Ternary 3-bi-immunity}
\label{3biimm}
Theorem~\ref{biimm_2} holds true also for ternary quantum random sequences, but a stronger result is true.
Informally, a sequence $\mathbf{x}\in A_{b}^{\omega}$ is $b$-bi-immune if for every $a\in A_b$, no algorithm can generate
infinitely many  pairs $(i, x_i=a)$ or  $(i, x_i\not= a)$. Formally, 
following~\cite{biimunity2020}, we say that a sequence $\mathbf{x}\in A_{b}^{\omega}$ is {\it $b$-bi-immune} if for every $a\in A_{b}$ the support  $\x^{-1}(a)=\{i\in\bbbn\mid x_i=a\}$  is
    bi-immune in the sense of computability theory~\cite{rogers1}, i.e.~the set and its complement contain no infinite computable subset. Obviously, $b$-bi-immunity is stronger than bi-immunity which is stronger than incomputability. 
    
 Consider a ternary sequence $\textbf{x}=x_1x_2\ldots \in A_3^{\omega}$ generated by the  QRNG.   
Then, for every $a\in A_3$ the set  $\x^{-1}(a)=\{i\in\bbbn\mid x_i=a\}$   and its complement contain no infinite computable subset because otherwise a definite value would need to be assigned to the observables corresponding to the measurement outputs contradicting the construction of the  QRNG (Theorem~\ref{EffecKS}). We have:

\begin{theorem}
\label{biimm_3}
 Assume the  Eigenstate and epr    principles. Then, every sequence generated by the QRNG is 3-bi-immune. 
\end{theorem}

It is  seen that the particular dimension 3 plays no role, so a stronger form of Theorem~\ref{biimm_2} is true:

\begin{theorem}
\label{biimm_b}
Assume the  Eigenstate and epr    principles. An infinite repetition of the
experiment 
measuring  a quantum value indefinite observable  in $\mathbb{C}^{b}$ always generates a
 $b$-bi-immune   sequence $\x\in A_b^{\omega}$. 
\end{theorem}

\subsection{Ternary unpredictability}
\label{3unpredict}

It is easy to check that the proof of
Theorem~\ref{2unpredict} works not only for the binary case, 
but for an   arbitrary alphabet $A_b$, $b\geq 2$. In particular we have

\begin{theorem}  
\label{bunpredict}
Assume the epr and Eigenstate  principles.
Let $\x$ be an infinite sequence obtained by measuring  a quantum value indefinite observable  in $\mathbb{C}^{b}$  in an infinite repetition of the experiment $E$. Then no single bit $x_i$ can be predicted.
\end{theorem}

\begin{corollary}
	Assume the epr and Eigenstate  principles.
Then, no single digit of every  sequence   $\x \in A_3^{\omega}$ generated by the QRNG  can be predicted. 
\end{corollary}

\section{Binary quantum random  sequences}
\label{2}

 As in most applications one needs
 binary random strings, in 
 this section we propose a method to transform ternary sequences into binary ones and, as in Section~\ref{3}, we study their bi-immunity, unpredictability and Borel normality.

\subsection{From ternary to binary sequences}
\label{3to2}

We give a simple  method to transform a ternary 
sequence into a binary sequence. The method is an alphabetics morphism $\varphi \colon A_3 \rightarrow A_2$

%

\begin{equation}
\label{amorph}
\varphi(a)= 
\begin{cases}0,&\text{if }a=0,
\\1,&\text{if }a=1,
\\0&\text{if } a=2,\end{cases}
\end{equation}

\noindent  which can be extended sequentially for strings, 
  $\mathbf{y}(n)=\varphi(\mathbf{x}(n))= \varphi(x_1)\varphi(x_2)\dots \varphi(x_n)$ and sequences $\mathbf{y}= \varphi(\mathbf{x}) = \varphi(x_1)\varphi(x_2)\dots \varphi(x_n) \dots $.

\subsection{Binary 2-bi-immunity}
\label{2biimm}


To prove 2-bi-immunity we use Theorem~\ref{biimm_3} and the following:

\begin{theorem}[\cite{biimunity2020}]
  Consider $b\ge 3$ and an alphabetic morphism $\varphi$ of $A_{b}$ onto
  $A_{b-1}$. Then for every $b$-bi-immune sequence $\x\in A_{b}^{\omega}$, the
  sequence $\varphi(\x)\in A_{b-1}^{\omega}$ is $(b-1)$-bi-immune.
\end{theorem}

\begin{corollary}
\label{ex.3to2}
The alphabetic morphism $\varphi$ defined by (\ref{amorph}) converts a
  $3$-bi-immune sequence into a $2$-bi-immune sequence.
\end{corollary}

\subsection{Binary unpredictability}

%
%

\begin{theorem}
	Assume the epr and Eigenstate  principles. Let $\mathbf{y}=\varphi (\mathbf{x})$, where $\mathbf{x}\in A_3^{\omega}$ is a ternary sequence generated by the  QRNG  and $\varphi$ is the alphabetic morphism defined in (\ref{amorph}).
Then, no single bit of    $\y \in A_2^{\omega}$  can be predicted. 
\end{theorem}
\begin{proof}
Let $\mathbf{y}$ be a sequence as in the statement above.
Fix an extractor $\langle\,  \rangle$, and
assume for the sake of contradiction that there exists a predictor $P_E$ for $\y$ which is $k,\langle \, \rangle$-correct for all $k\ge 1$.  Since $P_E$ \emph{never} makes an incorrect prediction, each of its predictions is correct with certainty, so the algorithm $P_E$ correctly and deterministically predicts the bits of $\y$, contradicting
 Corollary~\ref{ex.3to2}. A more physical explanation of this mathematical conclusion comes  from  the epr  principle:   $P_E$ predictions correspond to a value definite property of the system measured, i.e.~the QRNG, which contradicts  Theorem~\ref{biimm_2}. 
 \end{proof}

\subsection{Uniform distribution and Borel normality}
\label{2udnormal}

Recall that for $b\ge 2, A_b=\{0,1,2,\dots ,b-1\}$. Fix now an integer $m > 1$ and consider the alphabet $A_b^{m} =
\{a_{1}, \dots , a_{b^{m}}\} $ of all strings  $x \in  A_b^*$ with $|x|_b = m$, ordered lexicographically. A string  $x \in A_b^*$ will be denoted by   $x^m$ when we emphasise that  it belongs to $(A_b^{m})^*$.
 Take for example $A_2 = \{0,1\},
m=2,   A_2^{2} = \{00, 01, 10, 11\}$;  the string  $x =
0010101110 \e A_2^{\ast}$ will be denoted by  $x^2 = (00)(10)(10)(11)(10)$  when considered in $A_2^2$. Clearly,  $|x|_2 =  10$ and   $|x^2|_{4} = 5$.  
 In the same way a sequence  $\x \e A_b^{\omega}$  will be written  
 as  $\x^m$ when considered in $(A_b^m)^{\omega}$. 

Let  $\mathbf{x} \in  A_3^{\omega}$  and consider the random variable $X_n(\mathbf{x})=x_n$ on the probability space $( A_3^{\omega},\mathcal{B}(A_3^{\omega}),\mathbb{P}_3)$, where $\mathbb{P}_3$ is the  probability distribution of the QRNG. 
For simplicity we will write $X_n$ instead of $X_n(\mathbf{x})$ unless clarity suffers.
Then $X_1,X_2,\ldots,X_n, \dots $ is sequence of random variables mapping the sequence $\mathbf{x}$ to real-valued independent measurement outcomes, hence, it is a sequence of independent random variables with $\mathbb{P}_3(X_i=1)=\frac{1}{2}$ and $\mathbb{P}_3(X_i=0) = \mathbb{P}_3(X_i=2)=\frac{1}{4}$\rb . \ If $ \mathbf{x}\in A_3^{\omega}$, then $\mathbf{y}=\varphi(\mathbf{x})= \varphi(x_1)\varphi(x_2)\dots \in A_2^{\omega}$, so we can
 consider the random variable $
 Y_i(\mathbf{y})=y_i$.  Since the random variables $X_i$ correspond to independent events, we have that $\mathbb{P}_3(Y_i=1)=\mathbb{P}_3(X_i=1)=\frac{1}{2}$ and the expected value $\mathbb{E}_3(Y_i=0) = \mathbb{P}_3(X_i=0) + \mathbb{P}_3(X_i=2) = \frac{1}{2}$\rb.  Note that $Y_i$ takes values in $A_2$   with equal probabilities and $\mathbb{E}(Y_i)= 0 \cdot \mathbb{P}(Y_i=0) + 1 \cdot \mathbb{P}(Y_i=1) = \frac{1}{2}$\rb.  \   Thus $Y_1,Y_2,\ldots,Y_n, \dots $ is an independent and identically distributed (i.i.d.) sequence of random variables with uniform distribution, i.e.~in the Lebesgue probability space $( A_2^{\omega},\mathcal{B}(A_2^{\omega}),\mathbb{P})$.

 Is every sequence $\y$ Borel normal? To answer this question let's recall the definition of Borel normality. Let
$N_{i}(x)$ be the number of occurrences of  $i\in  A_b$
in the string $x \in  A_b^*$ and for
every $u\in A_b^{m}$
let $N_{u}^{m}(x^m)$ be
the number of occurrences of $u$ in the string $x^m \e (A_b^{m})^*$.
In the example above
 $ N_{0}^{1}(x)  =N_{1}^{1}(x) = 5$ and $N_{11}^{2}(x^2) = 1, N_{10}^{2}(x^2) = 3, N_{01}^{2}(x^2) = 0.$
There are strings $x \in  A_b^*$ for which   $x^m$ does not exist for some, even all,  $m$ (for example when $m$ is prime), but for all $m$ and $\x \e A_b^{\omega}$ the sequence $\x^m$ exists.

Recall that for $\x \e A_b^{\omega}$ and $n \geq 1$, $\x(n) = x_{1}x_{2} \dots
x_{n} \e A_b^{\ast}$.
The sequence $\x$ is called  $ m$-\emph{Borel normal} $(m \geq 1)$ in
case for every $u\in (A_b^{m})^*$ one has:
\[ \lim_{n \s \infty}
\frac{N_{u}^{m}(\x^m (\lfloor\frac{n}{m}\rfloor))}{\lfloor\frac{n}{m}\rfloor} = \frac{1}{b^{m}}\rb.\]

The sequence $\x \in A_b^{\omega} $ is called \emph{Borel normal} if
it is Borel $m$-normal, for every natural  $m \geq 1$.
  In particular,  a sequence $\x$ is Borel 1-normal when for every $a\in A_b$ we have:
 \[ \lim_{n \s \infty} \frac{N_{a}(\x(n))}{n} = \frac{1}{b}\rb.\]

We can generalise this construction of the i.i.d. random variables $(Y_i)$ by considering bit strings of arbitrary length $m\ge 1$ and then use the Strong Law of Large Numbers~\cite{Billingsley:1979aa} 
to get that {\it with probability 
one 
every bit sequence produced by the QRNG  is Borel normal.}
However, this result gives no new information as  Borel Law of Large Numbers~\cite{borel:09}  states that with probability one every bit sequence is Borel normal.
To get more insight we turn to  a finite version of Borel normality 
~\cite{DBLP:conf/dlt/Calude93} to analyse this property  for prefixes
of an arbitrary bit sequence produced via the  ternary sequence generated by the QRNG.

For every $\varepsilon >0$  and   integer $m > 1$  we say that a string $x\in A_2^* $
is {\it Borel normal with accuracy} $(m,\varepsilon)$ if

\begin{equation}
	\label{eps_normality}
	\left| \frac{N_{u}^{m}(x^m (\lfloor\frac{|x|_2}{m}\rfloor))}{\lfloor\frac{|x|_2}{m}\rfloor}  -  2^{-m} \right| \leq \varepsilon,	
\end{equation}
for each $u \in A_2^{m}$ and  $1 \leq m \leq \log_2  \log_2   |x|_2$.

It is useful to consider as  $\varepsilon$  a computable function of $|x|_2$ converging to zero when  $|x|_2$  to infinity. For example, in~\cite{DBLP:conf/dlt/Calude93,calude:02} the accuracy is 
$\sqrt{\frac{\log_2 |x|_2}{|x|_2}}$ and in~\cite{Abbott_2019} it is $\frac{1}{\log_2  |x|_2}$\rb.
Almost all algorithmic random strings of any length are Borel normal with these accuracies~\cite{DBLP:conf/dlt/Calude93,calude:02}.
Furthermore, {\it if all prefixes of a  bit sequence are  Borel normal, then the sequence itself is also  Borel normal.} 

\begin{lemma}
	\label{2normal_finite}

Let $\x\in A_2^{\omega}$ be a ternary sequence generated by the QRNG and let $\y=\varphi(\x)$. Then for every $m >  1$, the probability that $\y(m)$ is Borel normal with accuracy $\left(m,  \sqrt{\frac{\log_2 |x|_2}{|x|_2}}\right)$ is at least $1-\frac{1}{\sqrt{\log_2 m}}$\rb.
\end{lemma}
\begin{proof}
	Using~\cite[Lemma 5.43]{calude:02} we deduce  that for every $m > 1$, 
	\[ \#\left\{ z\in A_2^m \mid z \mbox{ is not Borel normal with accuracy }  
	\left(   m,  \sqrt{\frac{\log_2 |x|_2}{|x|_2}} \right) 
	\right\}
	\]
	\[ \leq \frac{2^m}{\sqrt{\log_2 m}}\rb,\]
	
\noindent hence the probability that $\y(m)$ is Borel normal with accuracy $\left(m,  \sqrt{\frac{\log_2 |x|_2}{|x|_2}}\right)$ is greater or equal to
\begin{equation}
\label{prob}
 1 -  \frac{1}{\sqrt{\log_2 m}}\rb.	
\end{equation} 	
\end{proof}

We note that the probability (\ref{prob}) increases with $m$ but this is not enough
to deduce that  $\y=\varphi(\x)$ is Borel normal: we only get  Borel normality with probability one. With larger and larger probabilities the prefixes of $\y$ are Borel normal, a property which is useful for practical purposes -- when only finitely many bits of  $\y$ can be computed -- and this property can be {\it tested} (and it was tested
in~\cite{PhysRevA.82.022102,DBLP:journals/corr/abs-1810-08718,Abbott_2019}).

\section{Conclusions}
\label{concl}
We have proposed a new ternary QRNG based on measuring located value indefinite observables   and prove that every sequence generated  is maximally unpredictable,  3-bi-immune (a stronger form of bi-immunity),  and its prefixes are Borel normal. 
 The  ternary quantum random digits produced by the   QRNG are algorithmically transformed  into quantum random bits using an alphabetic morphism which preserves all the above properties. One important question remains to be studied: how various forms of measurement error affect the properties of the quantum random bits obtained with this QRNG, see~\cite{Quantis2014,AbbottCalude10,way_theorem2012}. The QRNG proposed in this paper will be realised physically using a method similar to the one used in~\cite{PhysRevLett.119.240501} and the quality of randomness of samples of strings of length $2^{32}$ will be tested in comparison with  strings of pseudo-random bits, produced
  by the best available pseudo-random number generators, using various methods including those  in~\cite{Abbott_2019}.

One referee asked the following interesting question. Suppose a randomness test  rejects the hypothesis of randomness for many long strings of quantum random bits generated by the proposed QRNG. Does this fact  refute the corresponding physical theory on which the QRNG is based on? Such an approach may be attractive to physicists, because it is somewhat cheaper than other sophisticated 
precision experiments designed to test the validity of quantum mechanics.
Tentatively the answer is negative. Firstly, theoretically, that is, ignoring  a 
whole host of possibly erroneous hypotheses entering the empirical interpretation, every test of randomness applies to  finitely many, admittedly, very long, strings of quantum random bits, so it does not prove  non-randomness, which is an asymptotic property of the {\em infinite} sequences quantum random bits. 
Secondly, following~\cite{lakatosch}, we would check for a bug in the QRNG implementation and/or  some questionable/flawed  
assumptions implicitly made in its construction.
 Lastly, if no issues were found with the implementation and the test is failed in many cases on a large variety of very long strings obtained with different QRNGs based on the same theory, then  
 the theoretical assumptions made in Section~\ref{val_in} would be scrutinised.

\section*{Acknowledgement}
 This work was supported by the U.S. Office of Naval Research Global under Grant N62909-19-1-2038. The cooperation and support of S. Feng from the  Office of Naval Research Global,  N. Allen and K. Pudenz  at Lockheed Martin and A. Fedorov and his team at University of Queensland are  much appreciated.
 We also thank M. Dumitrescu, 
 L. Staiger, C. Stoica and, particularly
   K. Svozil,  for many useful discussions and suggestions.
   

\end{document}